\documentclass[journal]{IEEEtran}

\usepackage{amsmath,amsfonts,amsthm,amssymb,mathtools}
\usepackage{algorithmic}
\usepackage{algorithm}
\usepackage{array}
\usepackage{bm}
\usepackage{cite}
\usepackage{balance}
\usepackage{graphicx}
\usepackage{booktabs}
\usepackage[svgnames, dvipsnames]{xcolor}
\usepackage{hyperref}
\hypersetup{
colorlinks=false, 
linkcolor=cyan,
urlcolor=magenta,
citecolor=Mahogany,
}
\usepackage[caption=false,font=normalsize,labelfont=rm,textfont=rm]{subfig}
\usepackage{textcomp}
\usepackage{stfloats}
\usepackage{url}
\usepackage{verbatim}
\usepackage{physics}
\usepackage[most]{tcolorbox}  
\usepackage{lipsum}
\usepackage[normalem]{ulem}

\newtheorem{theorem}{Theorem} 

\newtheorem{lemma}[theorem]{Lemma}

\newtheorem{remark}{Remark} 

\begin{document}

\title{Age of Information Minimization in Goal-Oriented Communication with Processing and Cost of Actuation Error Constraints}

\author{Rishabh S. Pomaje, Jayanth S., Rajshekhar V. Bhat and Nikolaos Pappas%

\thanks{Rishabh S. Pomaje and Rajshekhar V. Bhat are with the Department of Electrical, Electronics, and Communication Engineering, Indian Institute of Technology Dharwad, Dharwad 580011, India (e-mail: \href{mailto:210020036@iitdh.ac.in}{210020036@iitdh.ac.in}; \href{mailto:rajshekhar.bhat@iitdh.ac.in}{rajshekhar.bhat@iitdh.ac.in}). 

Jayanth S. is currently with Nokia Solutions and Networks, Bengaluru, India (e-mail: \href{mailto:jayanth.s@nokia.com}{jayanth.s@nokia.com}). Nikolaos Pappas is with the Department of Computer and Information Science, Linköping University, Linköping 58183, Sweden (e-mail: \href{mailto:nikolaos.pappas@liu.se}{nikolaos.pappas@liu.se}). 

The work of Rajshekhar V. Bhat was supported by the Telecom Technology Development Fund (TTDF) under the TTDF scheme of the Department of Telecommunications (DoT), Government of India, Proposal ID: TTDF/6G/492, implemented through TCOE India. 

The work of N. Pappas has been supported in part by the Swedish Research Council (VR), ELLIIT, and the European Union (ETHER, 101096526, ROBUST-6G, 101139068, and 6G-LEADER, 101192080).}
}



\maketitle

\begin{abstract}
We study a goal-oriented communication system in which a source monitors an environment that evolves as a discrete-time, two-state Markov chain. At each time slot, a controller decides whether to sample the environment and if so whether to transmit a raw or processed sample, to the controller. Processing improves transmission reliability over an unreliable wireless channel, but incurs an additional cost. The objective is to minimize the long-term average age of information (AoI), subject to constraints on the costs incurred at the source and the cost of actuation error (CAE), a semantic metric that assigns different penalties to different actuation errors. Although reducing AoI can potentially help reduce CAE, optimizing AoI alone is insufficient, as it overlooks the evolution of the underlying process. For instance, faster source dynamics lead to higher CAE for the same average AoI, and different AoI trajectories can result in markedly different CAE under identical average AoI. To address this, we propose a stationary randomized policy that achieves an average AoI within a bounded multiplicative factor of the optimal among all feasible policies. Extensive numerical experiments are conducted to characterize system behavior under a range of parameters. These results offer insights into the feasibility of the optimization problem, the structure of near-optimal actions, and the fundamental trade-offs between AoI, CAE, and the costs involved. 
\end{abstract}
\begin{IEEEkeywords}
Age and semantics of information, cost of actuation error, stationary randomized policy, multiplicative performance bound. 
\end{IEEEkeywords}

\section{Introduction}
The timely and accurate delivery of information is crucial in real-time, mission-critical systems, including networked control systems, where a source that represents a physical process or device is often monitored remotely. This involves sensing and communicating measurements from the source to a monitor (typically a controller or decision-making unit), which then triggers actions at the same or different parts of the system based on the received information or the estimated state of the source.
The sensed or communicated raw data often need to be processed to infer the state of the source. Consider a simple scenario in which the processed data indicates whether the source is operating normally or abnormally, and an actuation needs to be performed for each state. In this case, a cost is associated with stale or inaccurate information at the monitor, as it can lead to incorrect or outdated actuation. This can be captured by the cost of actuation error (CAE) \cite{Pappas2021ICAS}. This cost is often asymmetric: staleness (i.e., the age of information (AoI) relative to the current state of the source) or inaccuracy when the source is in an abnormal state usually incurs significantly higher costs than similar imperfections during normal operation, though this depends on the specific application \cite{Luo2024semantic-aware,luo2025cost}. In this work, we consider the semantic property captured by CAE, as well as the staleness of information as measured by the AoI, to minimize the average AoI subject to constraints on the CAE and other constraints at the source related to transmission.

Traditional metrics such as latency do not adequately reflect data freshness, whereas AoI captures it by measuring the time elapsed since the last received update was generated \cite{kaul2012realtime}. By accounting for information staleness, AoI represents a significant step toward incorporating semantics. Comprehensive surveys on AoI and its variants are available in \cite{yates2021survey, kahraman2024iotsurvey, wang2023wirelesssurvey}, with research focusing on scheduling under various constraints and trade-offs with energy and transmission costs \cite{saurav2021cost, inan2021optimal, js2023distortion}. A range of optimization tools has been employed, including Markov decision processes (MDPs) and constrained MDPs \cite{agarwal2021energy, xu2024optimal, js2023distortion, ggb2021distortion, inan2021optimal, jayanth2023aopi}, Lyapunov techniques for online scheduling \cite{jayanth2023aopi, kadota2019scheduling, raikar2024reported}, competitive analysis for performance guarantees \cite{saurav2021cost, kadota2019scheduling}, and convex optimization for tractable policy design \cite{ggb2021distortion}.

Despite this extensive body of work, AoI captures only the timeliness of information, which is an innate semantic property, but neglects other contextual properties that are central to semantic communication \cite{Kountouris2021Semantics,Ayan2019AoIvsVoI}. In this context, both the timeliness and semantic relevance of information and their impact on actuation are critical \cite{Lan2021Semantic}, as in applications like autonomous vehicle control \cite{chen2025age} and real-time actuation from dynamic sources \cite{salimnejad2024realtime}, cyber-physical and autonomous networked systems \cite{Kountouris2021Semantics}. In this work, inspired by \cite{trevlakis2024toward, Luo2024semantic-aware, luo2025cost}, we consider AoI jointly with the CAE to capture both freshness and semantic relevance for actuation in a goal-oriented communication system. 

To capture the quality of information alongside AoI in goal-oriented communication systems, several works have examined the trade-off between AoI and distortion, often quantified as the mean squared error between the true value and the value received at the destination. While frequent sampling reduces AoI \cite{kadota2019scheduling}, compressing samples can further lower transmission delay and cost by sending fewer bits, at the expense of increased distortion \cite{ggb2021distortion, rajaraman2021aqi, bastopcu2019distortion}. In such settings, any received update, regardless of quality, contributes to reducing AoI. These studies demonstrate that tolerating higher distortion enables more frequent updates, revealing a fundamental trade-off between AoI and information quality. To jointly capture timeliness and content relevance, \cite{peng2024goal} introduces and optimizes the age of changed information.

In goal-oriented communication, prediction-based transmission is frequently employed. If future samples can be accurately estimated at the receiver, for instance based on the techniques in \cite{Chakravorty2017TAC}, transmissions can be skipped, saving both time and communication costs \cite{kam2020age, joshi2021minimization}. When a higher estimation error is tolerable, the receiver can treat its estimate as sufficiently close to the true source state for a longer duration, deferring updates until the error exceeds a predefined threshold. In this setting, updates are triggered only when the estimation error exceeds the acceptable bound, and AoI is considered to increase only beyond this threshold. Thus, a higher tolerable error leads to fewer transmissions and lower effective AoI, further highlighting the trade-off between freshness and semantic accuracy.

Many works have modeled sources as finite-state Markov chains in goal-oriented communication systems, where the receiver executes actuations based on its estimate of the source state \cite{salimnejad2024realtime,zakeri2025semantic,salimnejad2023state,gao2025goal}. In such settings, incorrect or delayed estimates can lead to incorrect or outdated actions, resulting in actuation errors. When no sample is received, the destination may rely on previous estimates or employ predictive estimators. Unlike standard estimation metrics such as mean squared error, these works quantify error as the cost of actuation error, reflecting the penalty or loss incurred due to incorrect or suboptimal actuation \cite{js2023distortion}. 

\subsection{Contributions}
While minimizing AoI through frequent transmissions can reduce state estimation error and thus actuation error, optimizing AoI alone may neglect the dynamics of the source process. Specifically, CAE captures a context attribute: for the same AoI, rapidly evolving sources induce higher CAE than slowly changing ones. Moreover, identical average AoI can lead to different CAE levels depending on the AoI trajectory. 

To address this, in this work, we consider a sensing system monitoring a discrete source, where AoI resets to a constant upon successful update and increases otherwise. We also track CAE, defined as the mismatch between the estimated state at the destination and the actual process state. Our goal is to optimize the average AoI subject to both CAE and transmission cost constraints. Available actions at the source include: 
(i) remain idle, (ii) sample the environment and transmit a raw sample, or (iii) sample, process (which incurs  additional processing cost but offers higher delivery reliability), and transmit an estimate of the state of the environment. This setup aligns with prior work on joint AoI-CAE optimization under resource constraints \cite{js2023distortion}. The main contributions of this work are as follows:
\begin{itemize}
    \item We formulate a joint control problem for sampling, processing, and transmission in a sensing-monitoring system, where the monitored environment evolves as a two-state discrete-time Markov chain. The optimization problem minimizes average AoI, subject to bounds on CAE and the cost incurred at the source, addressing the inadequacy of AoI alone in capturing the dynamics of fast- versus slow-changing environments.
    \item We propose a stationary randomized policy (SRP) that achieves an average AoI within a bounded multiplicative gap of two from the optimal policy over the feasible set. While the approach to deriving the performance bound is similar to prior work, we make several important adaptations to account for the different definition of AoI and the inclusion of CAE. 
    \item We conduct numerical experiments to evaluate the proposed policy under various system parameters and environment dynamics. Our results show that tight cost constraints lead to high AoI and infeasibility with strict CAE bounds. As the cost bound increases, more frequent sampling reduces both AoI and CAE. When processed transmission is affordable, the system favors it, but raw transmissions are preferred when the cost gap is large. Additionally, we observe that minimizing AoI does not always minimize CAE, and we highlight how channel reliability and source transition probabilities influence these metrics. Asymmetric transitions result in lower CAE due to longer dwell times in specific states.
\end{itemize}

    This paper is organized as follows: Section \ref{sec:system_model} presents the system model and formalizes the primary problem. Section \ref{sec:solution} investigates a class of SRPs as a solution approach. Section \ref{sec:num_analysis} provides the numerical evaluation and analysis. Finally, Section \ref{sec:conclusion} summarizes the findings and concludes the paper.

\section{The System Model}\label{sec:system_model}

    Consider a source (sensor node) tasked with monitoring an environment and providing timely and accurate updates to a destination (remote monitoring facility). Time is slotted, and at the beginning of each time slot \( n \in \{1, 2, 3, \dots, N\} \), the destination selects an action for the source to execute and communicates this decision to the source. The available actions are: (i) remain idle, (ii) sample the environment and transmit a raw sample, or (iii) sample, process, and transmit an estimate of the state of the environment.

    We assume that action selection, signaling from the destination to the source, sample acquisition and processing, transmission, and reception all occur within the same time slot. The system model is illustrated in Fig.~\ref{fig:system_model_complete}.

    In the following, we present the source and channel models, the formulations for CAE and AoI evolution, and the action-dependent transmission cost model.

    \begin{figure*}[!t]
        \centering
        \includegraphics[width=0.9\textwidth]{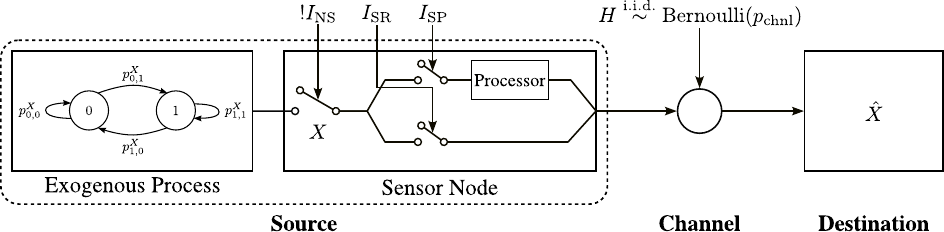}
        \caption{The source, directed by the destination, samples the stochastic process on demand and transmits data, processed or unprocessed, over a wireless channel.}
        \label{fig:system_model_complete}
    \end{figure*}

    \subsection{Source Model}
        In this work, we are interested in monitoring an environment that evolves as an exogenous two-state process modeled as a discrete-time, time-homogeneous Markov chain (DTMC), \( X(1), X(2), \ldots \). Let \( \mathcal{X} = \{0, 1\} \) be its state space, and let the probability of \( X(n) \) transitioning from State \( i \) to State \( j \) be \( p_{i, j}^X \) for any \( n \). We denote the stationary distribution of this two-state DTMC as \( \bm{\pi}_X = [\pi_0^X, \pi_1^X] \), where \( \pi_0^X = p_{1,0}^X /(p_{0,1}^X + p_{1,0}^X) \) and \( \pi_1^X = p_{0,1}^X/ (p_{0,1}^X + p_{1,0}^X) \).

    \subsection{Possible Actions at the Source}
        At each time slot, the source can take one of three possible actions: remain idle (\emph{no sample}, denoted as \( \mathrm{NS} \)), sample and transmit a \emph{raw} observation (\( \mathrm{SR} \)), or sample and \emph{process} the data before transmission (\( \mathrm{SP} \)). We define the action set as \( \mathcal{A} = \{\mathrm{NS}, \mathrm{SR}, \mathrm{SP}\} \) and introduce an indicator function for an action \( A \in \mathcal{A} \) at time slot \( n \) as follows:  
        \begin{align}  
            I_A(n) =  
            \begin{cases}  
                1, & \text{if action } A \text{ is taken during slot } n, \\  
                0, & \text{otherwise}.  
            \end{cases}  
        \end{align}  

        Each of the actions from the set \( \mathcal{A} \) has a cost\footnote{The cost associated with transmission actions is distinct from CAE. This cost, referred to as the transmission cost, arises from the operation of the transmitter and depends on the specific action taken—whether a sample is transmitted or not, and if transmitted, whether it is in raw or processed form. Thus, it may include the sampling and processing cost(s) implicitly.}\footnote{We consider the cost incurred at the  destination to be negligible in comparison to its resources.} associated with it. Let \( C(n) \) denote the real-time transmission cost function. Then, \( C(n) = \bm{c}^\top \bm{I}(n) \), where \( \bm{c} = [c_\mathrm{NS}, c_\mathrm{SR}, c_\mathrm{SP}] \) is the vector of costs associated with each of the possible actions, and \( \bm{I}(n) \) is the action vector for slot \( n \) defined as \(\bm{I}(n) = [I_\mathrm{NS}(n), I_\mathrm{SR}(n), I_\mathrm{SP}(n)]\). We consider the following constraint on the expected cost in long-term time-averaged transmission cost, denoted \( \bar{C} \), which is given by
        \[
            \bar{C} = \lim_{N \to \infty} \frac{1}{N} \sum_{n=1}^{N} \mathbb{E}[C(n)] \leq c_0,  
        \]
        where \( c_0 \) is the maximum allowed limit on \( \bar{C} \) and the expectation
        is with respect to the action vector. 
        \begin{remark}
        One can interpret \(c_\mathrm{NS}\) as the cost associated with keeping the system running, even when no sample is generated. The costs for sampling and transmitting the raw sample, or sampling, processing, and then transmitting, are captured by \(c_\mathrm{SR}\) and \(c_\mathrm{SP}\), respectively. While one could have set \(c_\mathrm{NS} = 0\) and subtracted this baseline cost from \(c_0\), the problem formulation would remain unchanged. However, we explicitly retain \(c_\mathrm{NS}\) to align with the probabilities \([p_\mathrm{NS}^\mathsf{R}, p_\mathrm{SR}^\mathsf{R}, p_\mathrm{SP}^\mathsf{R}]\), which correspond to not sampling, sampling and transmitting, and sampling, processing, and transmitting, respectively. These probabilities define the stationary randomized policy we adopt as a solution to the problem formulated at the start of the next section.
        \end{remark}

    \subsection{Channel Model}
        We consider a single-hop wireless channel that may be in either a good or bad state, where the probability of the channel being in a good state is \( p_\mathrm{chnl} \in [0,1] \). Let \( H(n) \) denote the channel state in slot \( n \), where \( H(n) \overset{\mathrm{i.i.d.}}{\sim} \text{Bernoulli}(p_\mathrm{chnl}) \). The probability of successful packet delivery, denoted as \( p^t_h \in [0,1] \), depends on the channel condition \( h \in \{0 \text{ (bad)}, 1 \text{ (good)}\} \) and the type of sample \( t \in \{r \text{ (raw)}, p \text{ (processed)}\} \).

    \subsection{Action Constraints and Probability of Successful Delivery}
        As defined before, if action vector at time \( n \) is denoted by \( \bm{I}(n) \), then any feasible decision policy must satisfy the following validity constraints: 
        \begin{align}  
            I_A(n) (I_A(n) - 1) &= 0, \quad \forall\ A \in \mathcal{A},\ n \in \mathbb{N}, \\  
            \sum_{A \in \mathcal{A}} I_A(n) &= 1, \quad \forall\ n \in \mathbb{N}.  
        \end{align}  

        Upon transmission of an update from the source, either the destination may receive the update successfully, or the delivery may fail. When a transmission is successful, the destination has perfect information or estimates about the state of the environment. When the delivery fails, the destination retains its previous information and uses it as the current estimate of the state of the environment.

        Let \( d(n) \) be a binary random variable that indicates successful packet delivery, where \( d(n) = 1 \) if the packet is successfully received and decoded at the destination and \( d(n) = 0 \) otherwise. The probability of a successful update depends on the action taken:  
        \begin{align}  
            d(n) & =   
            \begin{cases}  
                1, & \text{with probability } 0, \quad \text{if } \bm{I}(n) = [1,0,0], \\  
                1, & \text{with probability } \mu, \quad \text{if } \bm{I}(n) = [0,1,0], \\  
                1, & \text{with probability } \nu, \quad \text{if } \bm{I}(n) = [0,0,1],   
            \end{cases}  
        \end{align}  
        where \( \mu = p^r_1 p_\mathrm{chnl} + p^r_0 (1 - p_\mathrm{ chnl}) \) and \( \nu = p^p_1 p_\mathrm{chnl} + p^p_0 (1 - p_\mathrm{chnl}) \); \( \mu \) and \( \nu \) represent the success probability associated with the actions \( \mathrm{SR} \) and \( \mathrm{SP} \), respectively.

        Taking expectation with respect to the channel state randomness over \( d(n) \) conditioned on the action vector, we obtain:  
        \begin{align}  
            \mathbb{E}[d(n) \mid \bm{I}(n)] &= \mu I_\mathrm{SR}(n) + \nu I_\mathrm{SP}(n).  
        \end{align}  
        Taking expectation with respect to the randomness in the action vector and applying the law of iterated expectations, we have,
        \begin{align}  
            \psi(n) = \mathbb{E}[d(n)] &= \mu \mathbb{E}[I_\mathrm{SR}(n)] + \nu \mathbb{E}[I_\mathrm{SP}(n)].  \label{eq:psi_n}
        \end{align}  
        Note that \( \mu \) and \( \nu \) are determined by system parameters and \( \mathbb{E}[d(n)] \), which we denote as \( \psi(n) \) here on, represents the probability of a successful update, that is \( P(d(n) = 1) \), during time slot \( n \).

    \subsection{Cost of Actuation Error - A Context-Aware Distortion Measure}
        In many scenarios, including critical applications such as hazard detection, the destination performs some state-dependent action(s). It is essential that destination acquires updates from the source in timely and accurate manner (i.e., with low distortion). Decisions based on erroneous state estimates can lead to either unnecessary actuation or failure to act as and when required. Each of these outcomes is associated with a cost. For instance, actuation involves the use of physical resources, while taking no action can lead to damage or system failure. This necessitates consideration of  significance of the situation and context of the system into decision-making. 

        The CAE \cite{Fountoulakis2023COAE} quantifies the consequences and \emph{semantics}, formally capturing the cost incurred due to mismatch between the true state and its estimate. Let \( X(n) \) and \( \hat{X}(n) \) denote the true state and the estimate at time \( n \), respectively. We define the real-time\footnote{We use the term ``real-time'' to distinguish quantities measured at individual time slots from those averaged over time. This usage is conceptual and should not be interpreted in a strictly physical sense; it pertains to the timescale defined by discrete slots.} CAE denoted by \( \Delta(n) \) as:
        \[
            \Delta(n) = \Delta(X(n) = i, \hat{X}(n) = j) = \delta_{i,j}, \forall\  i, j \in \mathcal{X}, 
        \]  
        where \( \delta_{i, j} \geq 0 \) for \( i \neq j \) and \( \delta_{i, j} \leq 0 \) otherwise. The values \( \delta_{i,j} \) encode the actuation costs when the estimated state is \(j\) whereas the true state is \( i \). Notably, the CAE is non-commutative, meaning that \( \delta_{i, j} \neq \delta_{j, i} \) for \( i \neq j \) in general. This asymmetry encapsulates the directional nature of decision costs.  Consequently, the actuation error serves as a generalized, context-aware distortion measure. The long-term time-averaged expected CAE is given by  
        \[
            \bar{\Delta} = \lim_{N \to \infty} \frac{1}{N} \sum_{n=1}^{N} \mathbb{E}[\Delta(n)],
        \]
        where the expectation is with respect to the channel conditions and action vector.
        We consider the following constraint:  
        \[
            \bar{\Delta} \leq d_0,
        \]
        where \( d_0 \in \mathbb{R}_{++} \) denotes the bound on the long-term time-averaged expected CAE.

    \subsection{Age of Information}
        Let \( A(n) \) denote the AoI, at the destination, at the beginning of slot \( n \). The evolution of \( A(n) \) follows the update rule:  
        \begin{align}
            A(n + 1) = \begin{cases}
            w_0, & \text{if } \hat{X}(n) = 0, d(n) = 1,\\
            w_1, & \text{if } \hat{X}(n) = 1, d(n) = 1,\\
            A(n) + 1, & \text{otherwise}.
            \end{cases}
            \label{eq:age_of_information_definition}
        \end{align}
        Here, \( w_0, w_1 \in \mathbb{N} \) and we assume without loss of generality, \( w_1 \geq w_0 \). The definition above is adapted from \cite{js2023distortion}, where the goal was to bound the AoI for a specific state, meaning the AoI of a packet corresponding to that state was considered. We generalize this definition in \eqref{eq:age_of_information_definition}. By setting \(w_0 = 0\) and \(w_1 = 1\), we recover the original definition used in \cite{js2023distortion}. Many studies adopt a broader, non-state-discriminatory definition of AoI \cite{sinhaInfocomThroughput}. This definition can be obtained by setting \( w_0 = w_1 = 1 \), which is applicable to the problem addressed in this paper.

        Now, the long-term time-averaged expected AoI, denoted by \( \bar{A} \), is:
        \begin{align}
            \bar{A} = \lim_{N \rightarrow \infty} \frac{1}{N} \sum_{n=1}^{N} \mathbb{E}[A(n)],
        \end{align}
        where the expectation is taken with respect to the randomness in channel conditions and the action vector. 

\section{Optimization Analysis}\label{sec:solution}

    In this section we present the main optimization problem of this work, focus on the class of \emph{Stationary Randomized Policies (SRP)}, and derive the optimal policy within this class, which transforms the general problem into a simple and tractable structure, yielding a policy that is both analytically simple and practically implementable. We further establish a performance guarantee, showing that the AoI achieved under the optimal SRP is at most twice \( \mathrm{OPT}^\star \) in \eqref{eq:general_prob}. To facilitate this analysis, we first derive a closed-form expression for the real-time expected CAE, \(  \mathbb{E}[\Delta(n)] \), under an arbitrary feasible policy and then analyze the corresponding SRP.

    \subsection{Optimization Problem}
        In this work, we consider the following optimization problem:
        \begin{subequations}\label{eq:general_prob}
            \begin{align}
                \mathrm{OPT}^{\star} = \min_{\mathsf{P} \in \mathcal{P}} \quad & \lim_{N \to \infty} \frac{1}{N} \sum_{n=1}^{N} \mathbb{E}[A(n)] \label{subeq:general_prob_objective} 
                \\  
                \text{subject to} \quad & \bar{\Delta} \leq d_0,  \label{subeq:general_prob_distortion_constraint}
                \\  
                & \bar{C} \leq c_0, \label{subeq:general_prob_cost_constraint}
                \\  
                & I_{A}(n) (I_A(n) - 1) = 1, \forall\ A \in \mathcal{A},\ n \in \mathbb{N}, \label{subeq:general_problem_validity_constraint_01}
                \\  
                & \sum_{A \in \mathcal{A}} I_A(n) = 1, \forall\ n \in \mathbb{N}, \label{subeq:general_problem_validity_constraint_02}
            \end{align}
            \end{subequations}
        where \( \mathcal{P} \) denotes the class of all non-anticipative admissible policies, and \( \mathsf{P} \in \mathcal{P} \) is an admissible policy. The constraints \eqref{subeq:general_prob_distortion_constraint} and \eqref{subeq:general_prob_cost_constraint} enforce upper bounds on the long-term time-averaged expected CAE and transmission cost, respectively, while constraints \eqref{subeq:general_problem_validity_constraint_01} and \eqref{subeq:general_problem_validity_constraint_02} ensure the validity of the decision variables. The problem defined by \eqref{subeq:general_prob_objective}-\eqref{subeq:general_problem_validity_constraint_02} is referred to as the \emph{general problem}.

    \subsection{Real-Time Expected CAE}
        \begin{lemma}\label{lemma:inst_exp_dist}
            The real-time expected CAE is given by  
            \begin{align}
                \mathbb{E}[\Delta(n)] &=  \zeta + \xi \psi(n),  \label{eq:exp_distortion}
            \end{align}  
            where \( \zeta = \sum_{i, j} \delta_{i, j} \pi_i^X \pi_j^X \) and \( \xi = \pi_0^X \pi_1^X \sum_{i, j} (-1)^{i+j} \delta_{i, j} \). 
            \end{lemma}
        \begin{proof}
            See Appendix \ref{dist_expression}. 
        \end{proof}

        The expected real-time CAE \eqref{eq:exp_distortion} comprises two components: \( \zeta \) and \( \xi \psi(n) \). We make the following remarks regarding their significance.

        \begin{remark}
            The first term, \( \zeta \), represents the expected CAE under the assumption that the true state and its estimate evolve independently, with identical marginal distributions. Under this assumption, the CAE contribution is \( \delta_{i,j} \) whenever the true state is \( X = i\)  and the estimate is \( \hat{X} = j \).
        \end{remark}
        \begin{remark}
            However, the estimate is not independent of the true state—particularly in the case of a successful update, where the estimate becomes a deterministic function of the actual state. Among the four possible state-estimate pairs—\( (0, 0), (0, 1), (1, 0),\) and \( (1, 1) \)—a successful update corrects the mismatched pairs \( (0, 1) \) and \( (1, 0) \) with probability \( \psi(n) \). Specifically, if a pair \( (0, 1) \) is corrected to \( (0, 0)\), the CAE must be reduced by \( \delta_{0,1}\)  and increased by \( \delta_{0,0}\) . Similarly, correction of a \( (1, 0) \) pair to \( (1, 1) \) leads to an adjustment of \( \delta_{1,1} - \delta_{1,0} \). The second term, \( \xi \psi(n) \), captures this correction to the expected CAE arising from the dependence of the estimate on the true state. It aggregates the expected net reduction in CAE across such events, weighted by their likelihood under the update mechanism.  
        \end{remark}
        This result plays a key role---albeit subtly---by enabling a straightforward derivation of the performance bound for the SRP, later in the section. 

    \subsection{Stationary Randomized Policies (SRPs)}
        Let \(\mathcal{P}_R\) denote the set of all SRPs which choose an action with a certain probability in every slot. Concretely, a policy \(\mathsf{R} \in \mathcal{P}_R\) selects an action \( A \in \mathcal{A} \) with probability \( p_A^\mathsf{R} \) in slot \( n \). That is, the policy \(\mathsf{R}\) is fully characterized by the probability vector \(\bm{p}_\mathsf{R} = [p_\mathrm{NS}^\mathsf{R}, p_\mathrm{SR}^\mathsf{R}, p_\mathrm{SP}^\mathsf{R}]\).  

        In the rest of this section, we reformulate the primary general problem under the class of SRPs and present a reformulated convex optimization problem to obtain \( \bm{p}_\mathsf{R} \). We then demonstrate that the optimal AoI objective achieved under this optimal SRP is within a bounded multiplicative gap from that of the original general problem.

        \subsubsection{Reformulation}
            Under SRP \( \mathsf{R}\) we have, \( \mathbb{E}[I_A(n)] = p_A^\mathsf{R}, \forall\ A \in \mathcal{A}, n \in \mathbb{N} \). Substituting this into \eqref{eq:psi_n} yields:
            \begin{align}
                {\psi^\mathsf{R}(n)} = \psi^\mathsf{R} = \mu p_\mathrm{SR}^\mathsf{R} + \nu p_\mathrm{SP}^\mathsf{R}, 
            \end{align}
            and we can also show the following: 
            \begin{theorem}\label{th:aoi_expression_under_srp}
                Let \( \mathsf{R} \in \mathcal{P}_R\) be an arbitrary SRP. Under this policy, the long-term time-averaged expected AoI is given by:
                \begin{align}
                    \lim_{N \rightarrow \infty} \frac{1}{N} \sum_{n=1}^{N} \mathbb{E}[A(n)] = \left(\frac{1}{\psi^\mathsf{R}} - 1\right) + w_0 + \pi_1^X (w_1 - w_0). \label{eq:aoi_under_srp}
                \end{align}
            \end{theorem}
            \begin{proof}
                Refer Appendix \ref{srp_aoi} for a detailed proof. 
            \end{proof}

        We now consider the problem that one must solve in order to find the optimal SRP:
        \begin{subequations}
            \begin{align}
                \mathrm{OPT}_\mathsf{R}^{\star} = \min_{\mathsf{R} \in \mathcal{P}_R}\quad &\left(\frac{1}{\psi^\mathsf{R}} - 1\right) + w_0 + \pi_{1}^{X}\cdot (w_1 - w_0) \label{subeq:srp_age_objective}
                \\ 
                \text{subject to}\quad 
                &\zeta + \xi \psi^\mathsf{R} \leq d_0, \label{subeq:srp_distortion_constr}
                \\
                &\bm{c}^\top \bm{p}_\mathsf{R} \leq c_0, \label{subeq:srp_cost_constr} 
                \\
                & p_A^\mathsf{R} \geq 0,\quad \forall\ A \in \mathcal{A}. \label{subeq:srp_validity}
                \\
                &\sum_{A \in \mathcal{A}}p_{A}^\mathsf{R} = 1. \label{subeq:srp_normalization}
            \end{align}
            \label{eq:opt_prob_under_srp}
        \end{subequations}

        We refer to the problem \eqref{subeq:srp_age_objective}-\eqref{subeq:srp_normalization} as the \textit{SRP problem}. The constraints in \eqref{subeq:srp_distortion_constr} and \eqref{subeq:srp_cost_constr} follow from the time-invariant and stochastic nature of decision-making in SRPs. Additionally, the conditions in \eqref{subeq:srp_validity}-\eqref{subeq:srp_normalization} ensure that the SRP adheres to the validity constraints of the general problem. It can be verified that the above problem is convex. Consequently, it can be efficiently solved using standard convex optimization solvers.

        \subsubsection{Performance Bound}
            Policies in this class may be suboptimal. Therefore, we now establish a performance bound for this class of policies. To achieve this, we consider an alternative problem that serves as a lower bound for the general problem.
            \begin{subequations}
                \label{eq:lower_bound_problem}
                \begin{align}
                    L_B^\star = \min_{\mathsf{p} \in \mathcal{P}} \quad & \frac{1}{2} \left(\frac{1}{\hat{q}} - 1\right) + w_0 + \pi_1^{X} (w_1 - w_0) \label{subeq:lower_bound_objective}
                    \\[6pt]
                    \text{subject to} \quad & \bar{\Delta} \leq d_0, \label{subeq:lower_bound_dist_constr}
                    \\  
                    & \bar{C} \leq c_0, \label{subeq:lower_bound_cost_constr}
                    \\  
                    & I_{A}(n) (I_A(n) - 1) = 1, \forall\ A \in \mathcal{A},\ n \in \mathbb{N}, \label{subeq:lower_bound_validity_constr_01}
                    \\  
                    & \sum_{A \in \mathcal{A}} I_A(n) = 1, \forall\ n \in \mathbb{N}, \label{subeq:lower_bound_validity_constr_02}
                \end{align}
            \end{subequations}
            where \( \hat{q}\) denotes the long-term throughput of the source and defined as:
            \begin{align}
                \hat{q} = \lim_{N \rightarrow \infty} \frac{1}{N} \sum_{n=1}^{N} \mathbb{E}[d(n)].
                \label{eq:long_term_throughput}
            \end{align}
            The problem \eqref{subeq:lower_bound_objective}-\eqref{subeq:lower_bound_validity_constr_02} will be referred to as the \emph{lower-bound problem}. 
            \begin{lemma}\label{th:lowerbound}
                The lower bound problem \eqref{subeq:lower_bound_objective}-\eqref{subeq:lower_bound_validity_constr_02} is a lower bound of the general problem \eqref{subeq:general_prob_objective}-\eqref{subeq:general_problem_validity_constraint_02}. Mathematically, we have
                \begin{equation}
                    L_B^\star \leq \mathrm{OPT}^\star,
                \end{equation} 
                for every instance of our system model.
            \end{lemma}
            \begin{proof}
                Following an approach similar to that in \cite{sinhaInfocomThroughput}, we establish that  
                \[
                \lim_{N \rightarrow \infty} \frac{1}{N} \sum_{n=1}^{N}\mathbb{E}[A(n)] \geq \frac{1}{2} \left(\frac{1}{\hat{q}} - 1\right) + w_0 + \pi_1^{X} (w_1 - w_0).
                \]
                Due to the structural differences in the AoI evolution model, we employ a modified \emph{partially-expected sample path} approach rather than the single sample path approach used in \cite{sinhaInfocomThroughput}. A rigorous derivation of this lower bound on the AoI is presented in Appendix \ref{lower_bound_aoi}.    
                Substituting this lower bound into the objective function shows that the objective value of the original problem in \eqref{subeq:general_prob_objective}–\eqref{subeq:general_problem_validity_constraint_02} is lower bounded by that of the problem in \eqref{subeq:lower_bound_objective}–\eqref{subeq:lower_bound_validity_constr_02}.
            \end{proof}

            Taking advantage of the similarity between the objective functions in \eqref{subeq:srp_age_objective} and \eqref{subeq:lower_bound_objective}, we derive a performance bound, which is formally stated in the following theorem.

            \begin{theorem}{Optimality Bound on SRP class solutions.}{}
                For every instance of our system model, we have
                \[
                    \mathrm{OPT}_\mathsf{R}^\star < 2 \mathrm{OPT}^\star.
                \]
            \end{theorem} 
            \begin{proof}
                Let \( \mathsf{P}_{L_B} \) denote the policy that solves the problem of the lower bound. Let \( \bar{\Delta}_{\mathsf{P}_{L_B}} \text{ and } \bar{C}_{\mathsf{P}_{L_B}} \) be the long-term time average of the expected CAE and cost of actions associated with this policy. Construct a stationary randomized policy \( \mathsf{R}_0 \in \mathcal{P}_R \) such that,
                \begin{subequations}
                    \begin{align}
                        p_\mathrm{NS}^{\mathsf{R}_0} &= \lim_{N \rightarrow \infty} \frac{1}{N} \sum_{n=1}^{N} \mathbb{E}[I_\mathrm{NS}(n)]_{\mathsf{P}_{L_B}},  \\
                        p_\mathrm{SR}^{\mathsf{R}_0} &= \lim_{N \rightarrow \infty} \frac{1}{N} \sum_{n=1}^{N} \mathbb{E}[I_\mathrm{SR}(n)]_{\mathsf{P}_{L_B}},  
                    \end{align}
                    and 
                    \begin{align}
                        p_\mathrm{SP}^{\mathsf{R}_0} &= \lim_{N \rightarrow \infty} \frac{1}{N} \sum_{n=1}^{N} \mathbb{E}[I_\mathrm{SP}(n)]_{\mathsf{P}_{L_B}}. 
                    \end{align}
                \end{subequations}
                By exploiting the linearity of expectation, limits and summation, it is a  straightforward\footnote{In this step, we exploit the the closed-form expression for \( \mathbb{E}[\Delta(n)]\)  derived in Lemma~\ref{lemma:inst_exp_dist}.} exercise to show that \( \bar{\Delta}^{\mathsf{P}_{L_B}} = \bar{\Delta}^{\mathsf{R}_0}\)  and \( \bar{C}^{\mathsf{P}_{L_B}} = \bar{C}^{\mathsf{R}_0}\) . Thus, \( \mathsf{R}_0\)  satisfies all constraints if \( \mathsf{P}_{L_B}\)  is feasible. We conclude that \( \mathsf{R}_0\)  is feasible if \( \mathsf{P}_{L_B}\)  is feasible. Furthermore, by using the linearity of expectation, limits, and summation again, we can show that \( \hat{q}^{\mathsf{P}_{L_B}} = \psi^{\mathsf{R}_0}\) .

                Comparing the objectives,
                \begin{align}
                    \frac{\mathrm{OPT}_{\mathsf{R}_0}}{2} &= \frac{1}{2}\left(\frac{1}{\psi^{\mathsf{R}_0}} - 1\right) + \frac{w_0 + \pi_1^X (w_1 - w_0)}{2} \\
                    &< \frac{1}{2}\left(\frac{1}{\psi^{\mathsf{R}_0}} - 1\right) + w_0 + \pi_1^X (w_1 - w_0) \\
                    &= \frac{1}{2}\left(\frac{1}{\hat{q}^{\mathsf{P}_{L_B}}} - 1\right) + w_0 + \pi_1^X (w_1 - w_0) \\
                    &= L_B^\star.
                \end{align}
                Thus, \( \mathrm{OPT}_{\mathsf{R}_0} < 2 L_B^\star \). Recall that \( L_B^\star \leq \mathrm{OPT}^\star \leq \mathrm{OPT}^\star_\mathsf{R} \leq \mathrm{OPT}_{\mathsf{R}_0} \). Hence, 
                \begin{align}
                    \frac{\mathrm{OPT}^\star_\mathsf{R}}{\mathrm{OPT}^\star} \leq \frac{\mathrm{OPT_{\mathsf{R}_0}}}{\mathrm{OPT}^\star} \leq \frac{\mathrm{OPT}_{\mathsf{R}_0}}{L_B^\star} < 2.
                \end{align}
            \end{proof}
            Hence, \( {\mathrm{OPT}^\star_{\mathsf{R}}}/{\mathrm{OPT}^\star} < 2 \). Referring to the definition of optimality ratio given in \cite{sinhaInfocomThroughput}, we conclude that the SRP class of policies is \emph{2-optimal}. The last step is adopted from \cite{sinhaInfocomThroughput}. 

\section{Numerical Analysis}\label{sec:num_analysis}

    In this section, we numerically evaluate the impact of various system parameters on AoI and CAE under the SRP, including an assessment of problem feasibility. We further examine how these parameters influence the optimal action selection probabilities. Finally, we study the trade-off between AoI and CAE under the SRP. Unless mentioned otherwise, the model parameters were held constant as listed in \autoref{tab:parameters_numerical_analysis}.

    \begin{table}[t]
    \centering
    \renewcommand{\arraystretch}{1.2} 
    \caption{Parameter values used for numerical analysis.}
    \begin{tabular}{ll}
        \toprule
        \textbf{Parameter} & \textbf{Value} \\
        \midrule
        \( p_{0,1}^X, p_{1,0}^X \) & \( 0.35, 0.75 \) \\
        \( p_\mathrm{chnl} \) & \( 0.6 \) \\
        \(p^p_1, p^p_0, p^r_1, p^r_0\) & \( 0.9, 0.6, 0.75, 0.35 \) \\
        \( \delta_{0,0}, \delta_{0,1}, \delta_{1,0}, \delta_{1,1} \) & \(-0.25, 1, 1, -0.25 \) \\
        \( c_0, d_0 \) & \( 0.63, 0.2 \) \\
        \(c_\mathrm{NS}, c_\mathrm{SR}, c_\mathrm{SP}\) & \(0.1, 0.7, 1.2\) \\
        \(w_0, w_1\) & \(1, 1\) \\
        \bottomrule
    \end{tabular}
    \label{tab:parameters_numerical_analysis}
    \end{table} 

    \subsection{Feasibility Analysis}

        \begin{figure*}[!t]
            \centering
            \subfloat[]{%
                \includegraphics[width=\columnwidth]{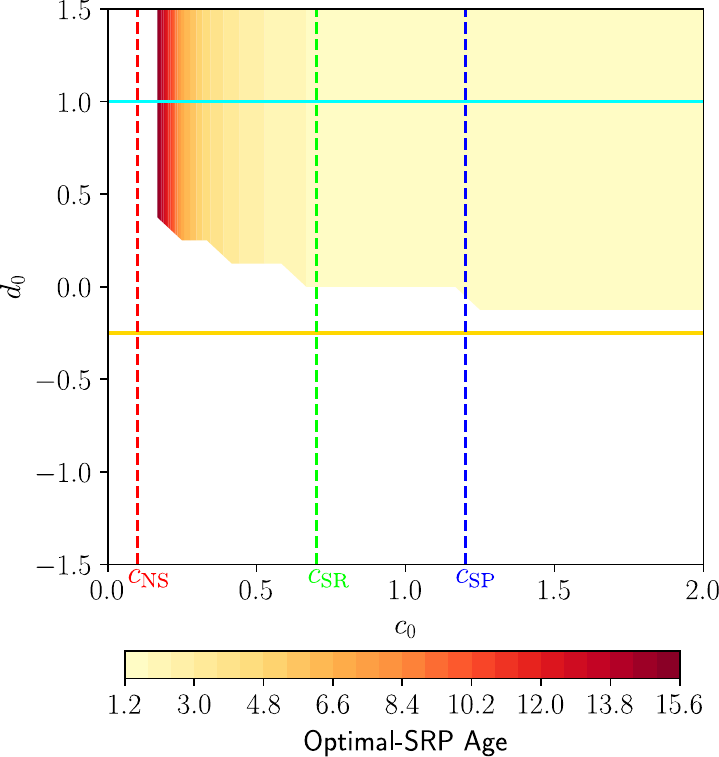}%
                \label{subfig:feas_age_plot}
            }
            \hfill
            \subfloat[]{%
                \includegraphics[width=\columnwidth]{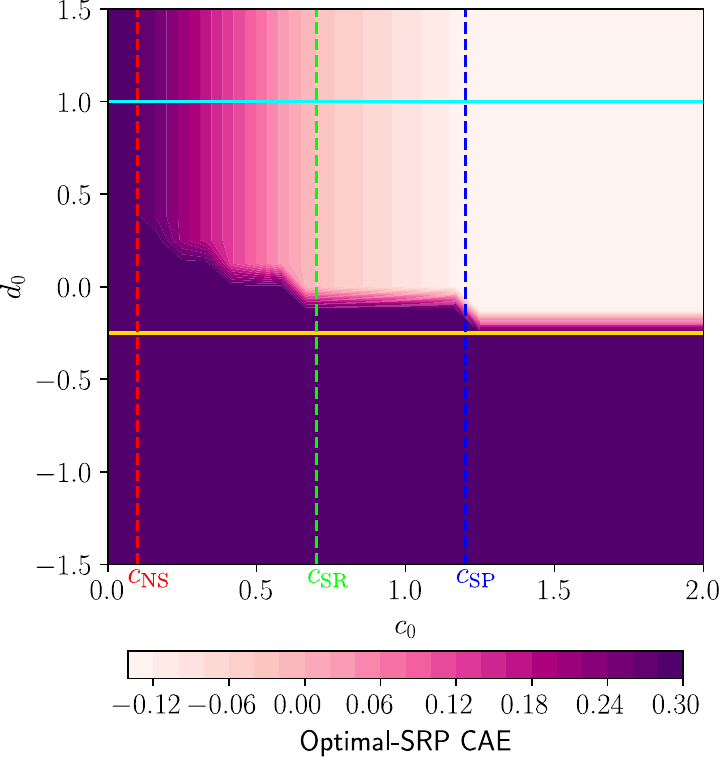}%
                \label{subfig:feas_dist_plot}
            }
            \caption{Variation of long-term averaged metrics under the optimal SRP with respect to constraint upper bounds: (\ref{subfig:feas_age_plot}) AoI, and (\ref{subfig:feas_dist_plot}) information CAE.}
            \label{fig:feasibility_01}
        \end{figure*}

        We study the set of constraint bounds \((c_0, d_0)\) for which the SRP optimization problem admits a solution, given a fixed set of parameters. Specifically, we vary \(c_0\) over \([0.0, 2.0]\) and \(d_0\) over \([-1.5, 1.5]\), while holding other parameters constant. The resulting optimal values of AoI and CAE are shown in Fig.~\ref{subfig:feas_age_plot} and Fig.~\ref{subfig:feas_dist_plot}, respectively. Ideally, the feasible region for the AoI should be convex. However, due to the discretization of \(c_0\) and \(d_0\) values, the lower boundary of the region exhibits fluctuations.

        For explainability, we partition each plot into three vertical segments—Regions A, B, and C—based on cost thresholds: Region A spans \(c_0 \in [c_\mathrm{NS}, c_\mathrm{SR})\), Region B spans \([c_\mathrm{SR}, c_\mathrm{SP})\), and Region C corresponds to \( c_0 \geq c_\mathrm{SP} \).

        In Region A, where \( c_0 \) is close to \( c_\mathrm{NS} \), the tight cost constraint limits sampling frequency, resulting in high AoI and CAE. Here, feasibility is primarily governed by the CAE constraint. For small \( d_0 \) values, the problem becomes infeasible as the system cannot meet strict CAE requirements under a limited budget. Increasing \( d_0 \) relaxes this constraint, restoring feasibility. However, AoI remains high due to persistent sampling limitations. As \( c_0 \) increases beyond \( c_\mathrm{SR} \), the system can sample more frequently, perhaps first enabling raw data transmission (Region B), and eventually processed data (Region C) once \( c_0 \ge c_\mathrm{SP} \). The latter benefits from higher delivery success due to reduced packet sizes.

        Notably, in the AoI plot, most contours are nearly horizontal—i.e., AoI remains largely unaffected by changes in \( d_0 \) for a fixed \(c_0\). This confirms that AoI is primarily influenced by the cost bound under the current setting. The CAE bound \(d_0\) has a limited impact on AoI, except near the infeasibility boundary where overly strict constraints render the problem unsolvable. Both the plots show a sharp transition in feasibility as \(d_0\) increases from negative values toward zero. This indicates that even slight relaxations in the constraints can significantly expand the feasible region, particularly under low-cost budgets.

    \subsection{Impact of Transmission Costs}

        \begin{figure*}[!t]
            \centering
            \subfloat[]{%
                \includegraphics[width=\columnwidth]{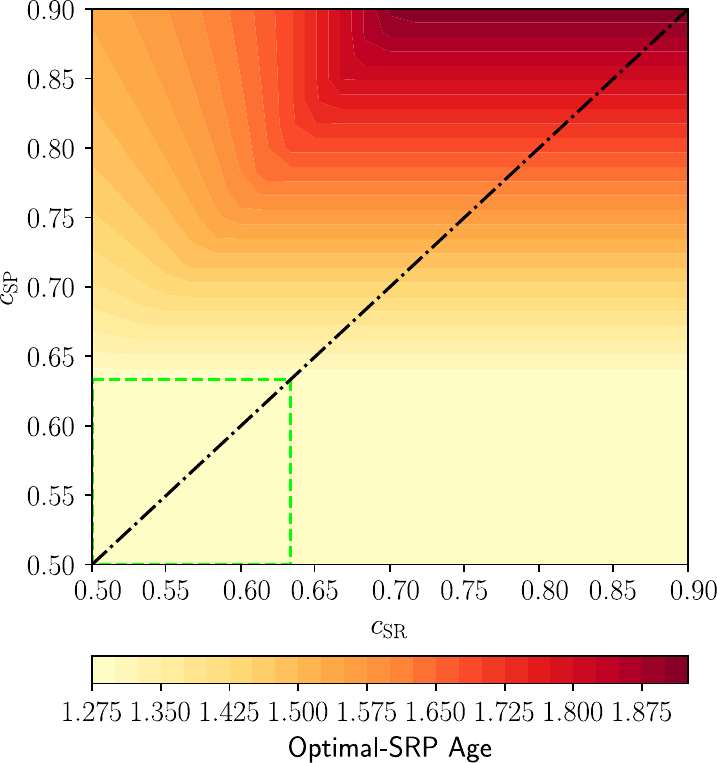}%
                \label{subfig:optimal_srp_age}
            }
            \hfill
            \subfloat[]{%
                \includegraphics[width=\columnwidth]{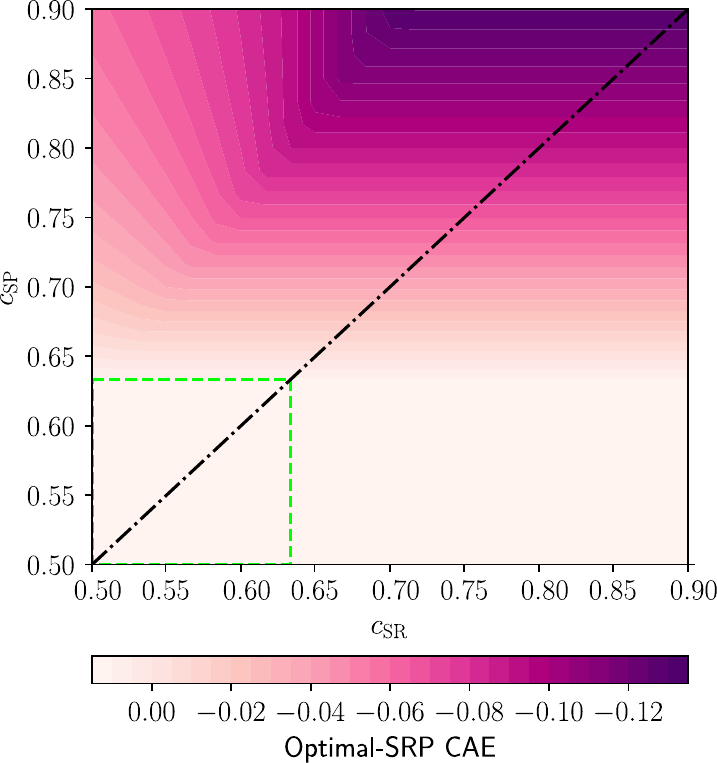}%
                \label{subfig:obs_dist_plot}
            }
            \caption{Impact of varying transmission costs on the observed long-term time-averaged metrics. The plot (\ref{subfig:optimal_srp_age}) illustrates the objective value, representing the long-term time-averaged AoI under the optimal SRP policy. The plot (\ref{subfig:obs_dist_plot}) shows the corresponding constraint value, depicting the CAE under the same policy.}
            \label{fig:metrics_vs_costs}
        \end{figure*}

        \begin{figure*}[!t]
            \centering
            \subfloat[]{%
                \includegraphics[width=0.32\textwidth]{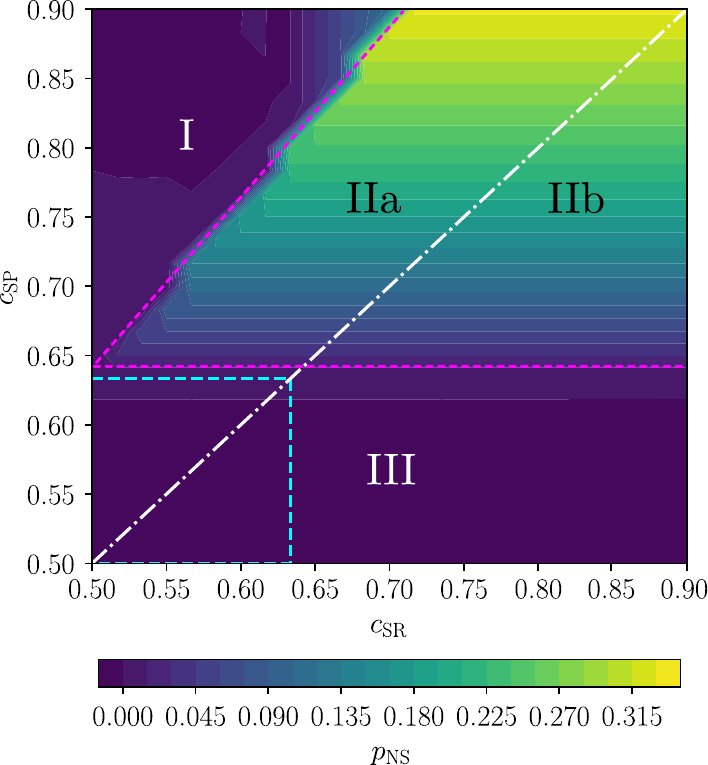}%
                \label{subfig:p_NS}
            }
            \hfill
            \subfloat[]{%
                \includegraphics[width=0.32\textwidth]{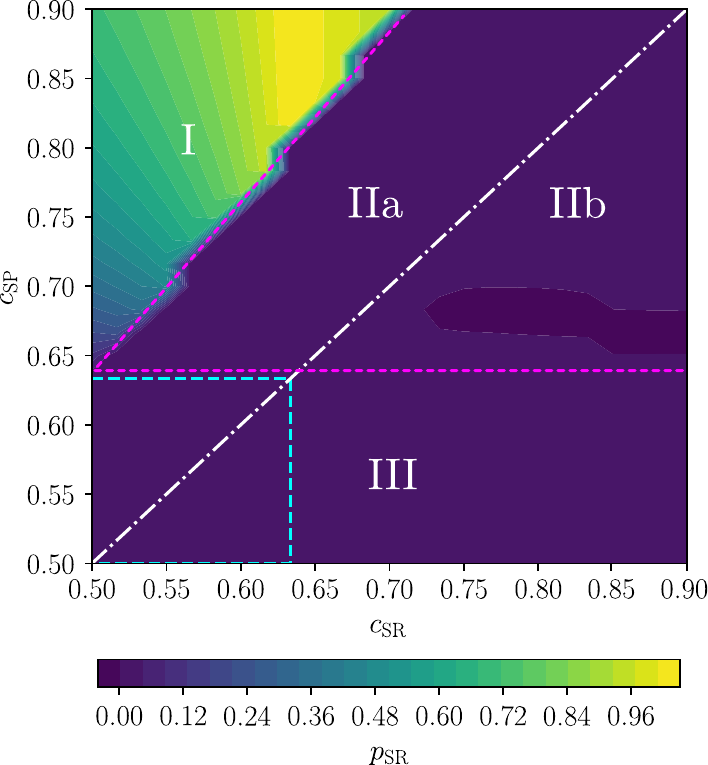}%
                \label{subfig:p_SR}
            }
            \hfill
            \subfloat[]{%
                \includegraphics[width=0.32\textwidth]{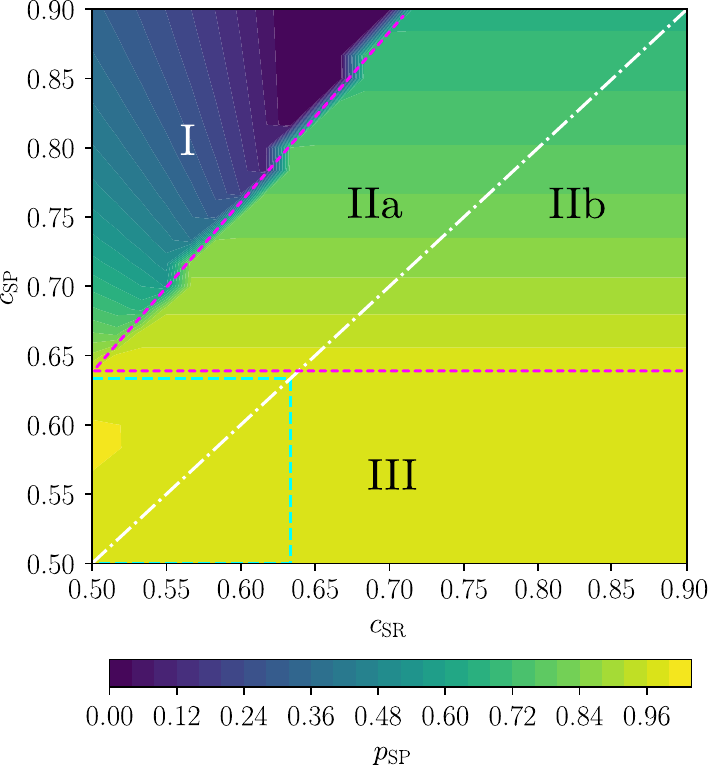}%
                \label{subfig:p_SP}
            }
            \caption{Effect of sampling costs on the optimal SRP policy. The optimal policy is characterized by the probabilities of different actions which are plotted as: (\ref{subfig:p_NS}) the probability of not sampling, (\ref{subfig:p_SR}) the probability of sampling and transmitting raw data, and (\ref{subfig:p_SP}) the probability of sampling and transmitting processed data (the state of the process).}
            \label{fig:probs_cost_plot}
        \end{figure*}

        We now analyze how variations in transmission costs affect the system's long-term performance and optimal behavior. Specifically, we fix \(c_\mathrm{NS} = 0.1\) and vary \(c_\mathrm{SR} \) and \( c_\mathrm{SP} \) over the range \( [0.5, 0.9] \), keeping the total cost constraint \( c_0 \) and other parameters fixed.

        Fig.~\ref{subfig:optimal_srp_age} and \ref{subfig:obs_dist_plot} present the observed long-term time-averaged AoI and CAE, respectively, under the optimal SRP. The corresponding policy consisting of the action probabilities \( p_\mathrm{NS}, p_\mathrm{SR} \) and \( p_\mathrm{SP} \) are shown in Fig.~\ref{subfig:p_NS}, \ref{subfig:p_SR}, and \ref{subfig:p_SP}, respectively. For ease of interpretation, all plots are segmented into three distinct and major regions (I, II, and III), each reflecting qualitatively different system behavior. 

        Each plot includes the following annotations: (i) a green/cyan dashed box indicating the region bounded by \(c_0\) on each axis, (ii) a black/white dash-dotted line corresponding to \( c_\mathrm{SR} = c_\mathrm{SP} \), and (iii) pink dotted lines demarcating the boundaries between the three identified regions. 

        In {Region III}, both AoI and CAE are significantly lower than in other regions. As illustrated in Fig.~\ref{fig:probs_cost_plot}, this can be attributed to the feasibility of executing the action \( \mathrm{SP}\) in almost every time slot, given that \( c_\mathrm{SP} \leq c_0\) . Additionally, the higher success probability associated with processed packets (\( \nu \) ) reduces both AoI and CAE. Consequently, the optimal SRP in this region predominantly favors the \( \mathrm{SP}\)  action.

        {Region II} can be further subdivided into subregions IIa (above the line \( c_\mathrm{SR} = c_\mathrm{SP}\) ) and IIb (below the line \( c_\mathrm{SR} = c_\mathrm{SP}\) ). In {Region IIb}, where \( c_\mathrm{SR} \geq c_\mathrm{SP}\) , processing before transmission is more cost-effective and offers higher success probability. However, since \( c_\mathrm{SP} > c_0\) , the cost constraint forces a reduction in sampling frequency, thereby increasing the probability of the no-sampling action (\( \mathrm{NS}\) ). This results in increased AoI and CAE. In {Region IIa}, although \( c_\mathrm{SP} > c_\mathrm{SR}\)  suggests a preference for the \( \mathrm{SR}\)  action, Fig.~\ref{fig:probs_cost_plot} shows a continued dominance of the \( \mathrm{SP}\)  action. This apparent contradiction highlights the role of differing success probabilities \( \mu\)  (for raw data) and \( \nu\)  (for processed data). A substantially higher \( \nu\)  can outweigh the cost advantage of the \( \mathrm{SR}\)  action, especially if \( \mu\)  is relatively low. As a result, the optimal SRP still favors \( \mathrm{SP}\) , alongside a relatively high \( p_\mathrm{NS}\)  compared to Region IIb.

        In {Region I}, where the cost asymmetry between raw and processed data transmission is most pronounced, the optimal SRP predominantly employs the \( \mathrm{SR}\)  action, occasionally interspersed with \( \mathrm{SP}\)  (close to \(c_0\)) and very rarely chooses \( \mathrm{NS}\) . This behavior is expected, as the system must compensate for the high cost of \( \mathrm{SP}\)  while still meeting the constraint. Two strategies are available: (i) transmit raw samples frequently and use processed data transmissions sparingly, or (ii) reduce sampling frequency. Since \( c_\mathrm{SR} \leq c_0\)  in this region, the former strategy becomes feasible. As a result, the optimal SRP in Region I delivers intermediate performance in terms of both AoI and CAE.

    \subsection{Tradeoffs Between AoI and CAE}

        \begin{figure}[!t] 
            \centering 
            \includegraphics[width=\columnwidth]{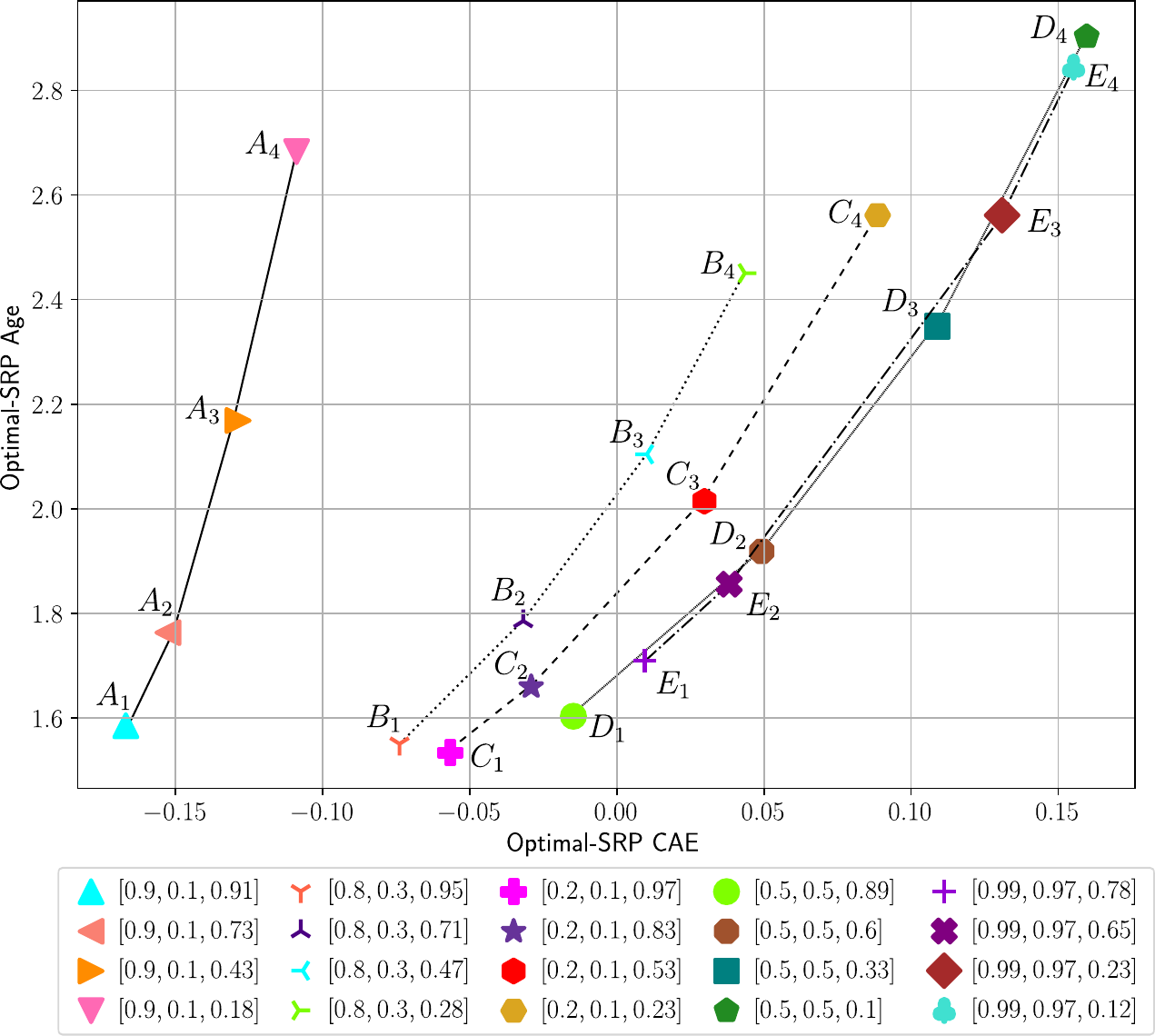} 
            \caption{The average Age of Information (AoI) and average CAE observed under the optimal SRP policy across various instances of the system model. Each legend entry denotes the tuple \( (p_{0,1}^X, p_{1,0}^h, p_{\mathrm{chnl}})\), corresponding to the parameters governing the exogenous process, hidden process, and channel reliability, respectively.} \label{fig:age_vs_dist} 
        \end{figure}

        Fig.~\ref{fig:age_vs_dist} illustrates the tradeoff between the average AoI and average CAE achieved by the optimal SRP under various system configurations. This analysis aims to understand how variations in the source parameters (transition probabilities) and channel reliability influence these performance metrics. Although Fig.~\ref{fig:metrics_vs_costs} suggests that AoI and CAE often increase or decrease together---potentially implying that minimizing one also minimizes the other---this is not always the case, as highlighted in Fig.~\ref{fig:age_vs_dist}.

        Consider the trajectories \( A_1 \rightarrow A_4\) , \( B_1 \rightarrow B_4\) , and \( C_1 \rightarrow C_4\)  in Fig.~\ref{fig:age_vs_dist}, which correspond to decreasing channel success probability \( p_{\mathrm{chnl}}\) . When the source transition probabilities are fixed (i.e., fixed \( (p_{0,1}^X, p_{1,0}^X)\) ), both AoI and CAE increase monotonically as channel reliability degrades. This is intuitive---fewer successful transmissions lead to staler information and less accurate estimates.

        However, consider points \( A_2, B_2 \), and \( E_2\), which have nearly identical \( p_{\mathrm{chnl}}\) . These points exhibit similar AoI but significantly different CAE levels. Since the channel reliability is constant, the variation in CAE arises from differences in the source transition probabilities. Specifically, when the source transitions are asymmetric (e.g., one high and one low), the system tends to dwell longer in the state with the lower transition probability. Because the estimate is typically held constant between updates, the likelihood of matching the actual state is higher, resulting in lower CAE. In contrast, when both transition probabilities are high (e.g., \( (p_{0,1}^X, p_{1,0}^X) = (0.5, 0.5)\)  or \( (0.9, 0.9)\) ), the state changes frequently, increasing the risk of mismatch and hence CAE.

        In summary, this figure emphasizes that AoI and CAE capture fundamentally different dimensions of system performance. Optimizing one may not guarantee optimality in the other. For systems where both information freshness and accuracy matter, sampling and transmission strategies must be designed with joint consideration of both metrics.
        
\section{Conclusions}\label{sec:conclusion}
    We studied a sensing-monitoring system in which decisions on sampling and transmission—whether raw or processed—are made under transmission cost and CAE constraints. Processed samples offer a higher success probability for update delivery but incur a greater cost. We proposed a stationary randomized policy that selects actions with fixed probabilities at each time instant to minimize average AoI while satisfying both constraints. We analytically showed that the average AoI achieved under this policy is within a factor of two of the optimal value over all feasible policies.

    Our results showed that AoI is primarily influenced by the cost bound, while CAE depends on the temporal dynamics of the underlying Markov process for a given average AoI. The feasibility analysis revealed that under tight cost constraints, the system often becomes infeasible when CAE bounds are strict, and AoI remains high due to limited sampling opportunities. As the cost bound increases, the system supports more frequent sampling and enables raw/processed transmissions, thereby reducing both AoI and CAE. We also analyzed how changes in transmission costs affect long-term behavior. When processed transmission becomes affordable, the system prefers it due to its higher success probability, resulting in low AoI and CAE. In contrast, when the cost gap between raw and processed transmissions is large, the system favors raw transmissions, yielding moderate improvements in AoI and CAE. Finally, we analyzed the trade-off between AoI and CAE, which highlighted that minimizing AoI does not necessarily minimize CAE. Our results demonstrated that both channel reliability and source transition dynamics significantly influence these metrics. Specifically, we observed that asymmetric source transitions lead to lower CAE due to longer dwell times in certain states, which increases the likelihood that copying the previous state as the estimate when no transmission occurs yields a correct prediction.

    In this work, we focused on a single-sensor setup as a foundational case. Future directions include extending the framework to multi-state source models, accommodating multiple sources with contention, developing adaptive policies tailored to dynamic environments, and deriving tighter analytical performance bounds to deepen theoretical understanding.

\begin{appendices}

\section{Proof of Lemma~\ref{lemma:inst_exp_dist}}\label{dist_expression}

    Let \( S(n) \) denote the event of a successful update occurring in time slot \( n \), i.e., \( d(n) = 1 \) and \( \bar{S}(n) \) be the complementary event.
    We begin by deriving the marginal distribution of the estimated state, specifically, \(
    P_{\hat{X}(n)}(0) = P(\hat{X}(n) = 0)\) and \( P_{\hat{X}(n)}(1) = P(\hat{X}(n) = 1)\).  
    Using the law of total probability, we express \( P_{\hat{X}(n)}(0) \) as  
    \begin{multline}
        \label{eq:est_recursive_n}
        P_{\hat{X}(n)}(0) = P(\hat{X}(n) = 0 \mid S(n)) P(S(n))\\ 
        + P(\hat{X}(n) = 0 \mid \bar{S}(n)) P(\bar{S}(n)).
    \end{multline}  

    A successful update leading to \( \hat{X}(n) = 0 \) is possible only if the true state was \( X(n) = 0 \), implying that \( P(\hat{X}(n) = 0 \mid S(n)) = P(X(n) = 0) = \pi_0^X\) . On the other hand, in the absence of an update (or if the update attempt fails), the estimated state remains unchanged, meaning that \( P(\hat{X}(n) = 0 \mid \bar{S}(n)) = P(\hat{X}({n-1}) = 0)\) . Given that the probabilities of a successful and unsuccessful update are \( P(S(n)) = \psi(n) \) and \( P(\bar{S}(n)) = 1 - \psi(n) \), respectively, substitution in \eqref{eq:est_recursive_n} yields the recursive expression:  
    \begin{align}
        P_{\hat{X}(n)}(0) = \pi_0^X \psi(n) + P_{\hat{X}({n-1})}(0) (1 - \psi(n)).
    \end{align}

    For tractability and simplicity of presentation, we set the initial (priori) distribution of the estimate \( \hat{X} \) equal to that of the true state \( X \), i.e., we consider that  
    \[
        \hat{X}(0) \sim \text{Bernoulli}(\pi_1^X).
    \]  
    To examine the evolution of the estimate over time, we consider the first few iterations of the recursive equation.  
    For \( n = 1 \), substituting \( P_{\hat{X}(0)}(0) = \pi_0^X \) into the update equation, we obtain  
    \begin{align}
        P_{\hat{X}(1)}(0) = \pi_0^X \psi(1) + \pi_0^X (1 - \psi(1)) = \pi_0^X.
    \end{align}  
    Similarly, for \( n = 2 \), using \( P_{\hat{X}(1)}(0) = \pi_0^X \), we get  
    \begin{align}
        P_{\hat{X}(2)}(0) = \pi_0^X \psi(2) + \pi_0^X (1 - \psi(2)) = \pi_0^X.
    \end{align}  
    By induction, it follows that for all \( n \), the probability remains unchanged, \( P_{\hat{X}(n)}(0) = \pi_0^X,\text{ and }  P_{\hat{X}(n)}(1) = \pi_1^X\) . Thus, the marginal distribution of \( \hat{X}(n) \) is identical to the prior distribution of \( X(n) \). 

    To derive the joint distribution of the true state \( X(n) \) and the estimate \( \hat{X}(n) \), we consider  
    \[
        P_{X(n), \hat{X}(n)}(i, j) = P(X(n) = i, \hat{X}(n) = j),\ \forall\ i, j \in \mathcal{X}.
    \]  
    Applying the law of total probability, 
    \begin{multline}
        P_{X(n),\hat{X}(n)}(0, 0) = P_{X(n),\hat{X}(n)}(0, 0 \mid S(n)) P(S(n)) \\  
        + P_{X(n),\hat{X}(n)}(0, 0 \mid \bar{S}(n)) P(\bar{S}(n)).
        \label{eq:joint_dist_00_n}
    \end{multline}  
    By applying the definition of conditional probabilities, we have
    \begin{align}
        P_{X(n),\hat{X}(n)}(0, 0) &= P_{\hat{X}(n)}(0\mid X(n) = 0, S(n)) \nonumber\\
        &\quad\quad \times P_{X(n)}(0 \mid S(n)) P(S(n)) \nonumber \\  
        &\quad + P_{\hat{X}(n) }(0 \mid X(n) = 0, \bar{S}(n))  \nonumber\\
        &\quad\quad \times  P_{X(n)}(0 \mid \bar{S}(n)) P(\bar{S}(n)).
    \end{align}  

    Since a successful update ensures \( \hat{X}(n) \) aligns with \( X(n) \), we have \(
    P_{\hat{X}(n)}(0 \mid X(n) = 0, S(n)) = 1.\)  
    Further, if no update occurs (\( \bar{S}(n) \)), the estimate remains unchanged, leading to \(
    P_{\hat{X}(n)}(0 \mid X(n) = 0, \bar{S}(n)) = P_{\hat{X}({n-1})}(0).
    \)  
    Substituting into \eqref{eq:joint_dist_00_n},  
    \begin{align}
        P_{X(n), \hat{X}(n)}(0, 0) = (\pi_0^X)^2 + \pi_0^X \pi_1^X \psi(n).
    \end{align}  
    By symmetry, interchanging 0 with 1 yields  
    \begin{align}
        P_{X(n), \hat{X}(n)}(1, 1) &= (\pi_1^X)^2 + \pi_0^X \pi_1^X \psi(n).
    \end{align}  

    Consider the asymmetric case where the joint probability distribution of the true state and its estimate is given by  
    \begin{align}
        P_{X(n),\hat{X}(n)}(0, 1) &= P_{\hat{X}(n) }(1 \mid X(n) = 0, S(n)) \nonumber\\
        &\quad\quad \times P_{X(n)}(0 \mid S(n)) P(S(n)) \nonumber \\  
        &\quad + P_{\hat{X}(n)}(1 \mid X(n) = 0, \bar{S}(n)) \nonumber\\
        &\quad\quad \times P_{X(n)}(0 \mid \bar{S}(n)) P(\bar{S}(n)).  
        \label{eq:joint_dist_01_raw}  
    \end{align}
    Since the estimate must match the true state in a given time slot under a successful update, it follows that \( P(\hat{X}(n) = 1 \mid X(n) = 0, S(n)) = 0 \). Additionally, by previously established argument, \( P(\hat{X}(n) = 1 \mid X(n) = 0, \bar{S}(n)) = P_{\hat{X}({n-1})}(1) \). Substituting these into \eqref{eq:joint_dist_01_raw}, we obtain  
    \begin{align}
        P_{X(n), \hat{X}(n)}(0, 1) &= P_{\hat{X}({n-1})}(1) \pi_0^X (1 - \psi(n)) \\
        &= \pi_1^X \pi_0^X (1 - \psi(n)).  
    \end{align}  
    By symmetry, interchanging 0 and 1 yields  
    \begin{align}
        P_{X(n), \hat{X}(n)}(1, 0) &= \pi_0^X \pi_1^X (1 - \psi(n)).  
    \end{align}  

    Using the joint distribution of the true state and its estimate, the real-time expected CAE is given by  
    \begin{align}
        \mathbb{E}[\Delta(n)] &= \sum_{i, j} \delta_{i, j} \pi_i^X \pi_j^X + \psi(n) \pi_0^X \pi_1^X \sum_{i, j} (-1)^{i+j} \delta_{i, j}. \label{eq:expected_distortion}
    \end{align}  

    Notably, in this expression, the only policy-dependent term is \( \psi(n) \), while the remaining terms consist of fixed parameters of the model instance. Thus, the real-time expected CAE is given by the expression \( \zeta + \xi \psi(n)\) , where \( \zeta = \sum_{i, j} \delta_{i, j} \pi_i^X \pi_j^X \)  and \( \xi = \pi_0^X \pi_1^X \sum_{i, j} (-1)^{i+j} \delta_{i, j}\)  are known constants. 

\section{Proof of Theorem~\ref{th:aoi_expression_under_srp}}\label{srp_aoi}
    The evolution of the AoI can be modeled as sample paths of an independent DTMC. We consider \(1 \leq w_0 < w_1\). The state space of this DTMC is given by 
    \( \Omega_A = \{w_0, w_0 + 1, \dots, w_1 - 1, w_1, w_1+1, \dots\}\) . Under the SRP \(\mathsf{R}\), the transition probabilities are independent of the current state. We derive the stationary distribution of the resulting DTMC and subsequently compute its expected state value.

    Recall that \( S(n) \) denotes the event of a successful update occurring in time slot \( n \), i.e., \( d(n) = 1 \). For any \( s \in \Omega_A \), we have \( P(A(n+1) = w_0 \mid A(n) = s) \overset{\text{(a)}}{=} P(X(n)=0, S(n))  \overset{\text{(b)}}{=} P(X(n) = 0) P(S(n))\) , where (a) follows from the update rule of AoI, and (b) follows from the independence of \( X(n) \) and \( S(n) \). Since it is known that \( P_{X(n)}(0) = \pi_{0}^{X} \) and \( P(S(n)) = {\psi^\mathsf{R}(n)} = \psi^\mathsf{R} \), we have: \( P(A(n+1) = w_0 \mid A(n) = s) = \pi_{0}^{X} \psi^\mathsf{R}\) . Similarly, it follows that:  \( P(A(n+1) = w_1 \mid A(n) = s) = \pi_{1}^{X} \psi^\mathsf{R}\) , {except when \( s = w_1 - 1 \)}. Using these transition probabilities, define: \( p_0 \triangleq P(A(n+1) = w_0 \mid {A(n)} = s)\) , \( p_1 \triangleq P(A(n+1) = w_1 \mid A(n) = s)\) , and \( p_2 \triangleq P(A(n+1) = s+1 \mid A(n) = s)\) .  
    Since the transition probabilities associated with a state must sum to one, we obtain \( p_0 + p_1 + p_2 = 1 \), and thus \( p_2 = 1 - p_0 - p_1 = 1 - \pi_{0}^{X} \psi^\mathsf{R} - \pi_{1}^{X} \psi^\mathsf{R} = 1 - \big(\pi_{0}^{X} + \pi_{1}^{X} \big) \psi^\mathsf{R} = 1 - \psi^\mathsf{R} \).
    Notably, \( p_2 \) represents the probability that the AoI increments by one unit, which occurs either when no update is attempted or when an attempted update fails. The simplification \( p_2 = 1 - \psi^\mathsf{R} \), where \( \psi^\mathsf{R} \) denotes the probability of a successful update, i.e., \( P(d(n) = 1) \), is intuitive.
    \begin{figure*}[t]
        \centering
        \includegraphics[width=0.75\textwidth]{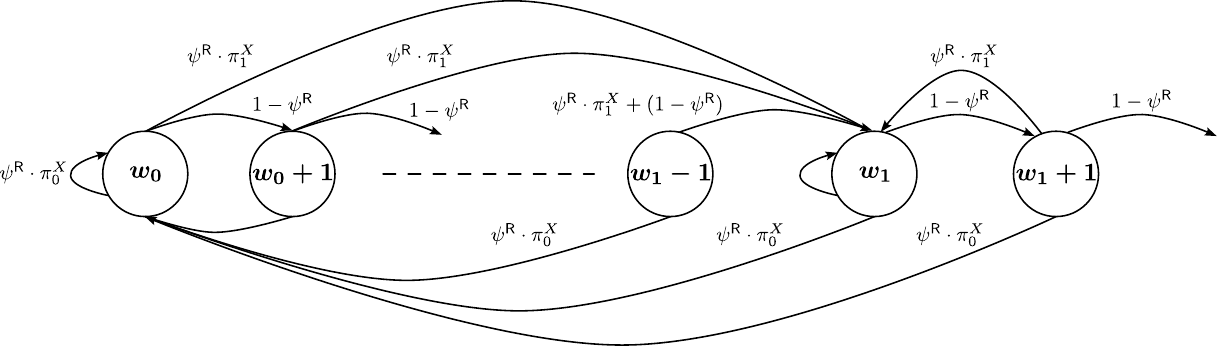}
        \caption{The state transition diagram of the DTMC used to derive the long-term time-averaged expected AoI.}
        \label{fig:aoi_dtmc}
    \end{figure*}
    Using the transition probabilities and the state transition diagram (see Fig.~\ref{fig:aoi_dtmc}), the corresponding state transition matrix, denoted by \( \bm{P}_A\) , can be constructed:
    \begin{equation}
        \resizebox{\columnwidth}{!}{
        \( \bm{P}_A =
        \begin{bmatrix}
            \pi_0^{X} \psi^\mathsf{R} & 1 - \psi^\mathsf{R} & 0  & \cdots & 0 & \pi_1^{X}  \psi^\mathsf{R} & 0 & \cdots \\
            \pi_0^{X} \psi^\mathsf{R} & 0 & 1 - \psi^\mathsf{R} & \cdots & 0 & \pi_1^{X}  \psi^\mathsf{R} & 0 & \cdots \\
            \vdots & \vdots & \vdots & \cdots & \vdots & \vdots & \vdots & \cdots \\
            \pi_0^{X} \psi^\mathsf{R} & 0 & 0 & \cdots & 1 - \psi^\mathsf{R} & \pi_1^{X} \psi^\mathsf{R} & 0 & \cdots \\
            \pi_0^{X} \psi^\mathsf{R} & 0 & 0 & \cdots & 0 & 1 - \psi^\mathsf{R} + \pi_1^{X} \psi^\mathsf{R} & 0 & \cdots \\
            \pi_0^{X} \psi^\mathsf{R} & 0 & 0 & \cdots & 0 & \pi_1^{X} \psi^\mathsf{R} & 1 - \psi^\mathsf{R} & \cdots \\
            \pi_0^{X} \psi^\mathsf{R} & 0 & 0 & \cdots & 0 & \pi_1^{X} \psi^\mathsf{R} & 0 & \cdots \\
            \vdots & \vdots & \vdots & \cdots & \vdots & \vdots & \vdots & \cdots
        \end{bmatrix}\) .}
    \end{equation}
    Using the normalization condition \( \sum_{s = w_0}^{\infty} \pi_{s}^{A} = 1 \), we determine the stationary distribution \( \bm{\pi}_A \) by solving the balance equation \( \bm{\pi}_A = \bm{\pi}_A \bm{P}_A \): 
    \begin{equation}
        \pi_k^A = 
        \begin{cases}
            \pi_0^X \psi^\mathsf{R} (1 - \psi^\mathsf{R})^{k - w_0}, & \text{for } w_0 \leq k < w_1, \\
            \begin{aligned}[t]
                &\pi_0^X \psi^\mathsf{R} (1 - \psi^\mathsf{R})^{k - w_0} \\
                &\quad + \pi_1^X \psi^\mathsf{R} (1 - \psi^\mathsf{R})^{k - w_1},
            \end{aligned}
            & \text{for } k \geq w_1.
        \end{cases}
    \end{equation}
    Using the stationary distribution, the expected real-time AoI is given by \( \mathbb{E}[A(n)] = \sum_{k=w_0}^{\infty} k \pi_k^{A} \).

    Splitting the summation at \( w_1 \), where the stationary distribution expression changes, and substituting the respective expressions, the expected AoI is given by 
    \( \mathbb{E}[A(n)] = \sum_{k=w_0}^{w_1-1} k \pi_0^{X} \psi^\mathsf{R} (1 - \psi^\mathsf{R})^{k-w_0}  + \sum_{k=w_1}^{\infty} k \pi_0^{X} \psi^\mathsf{R} (1 - \psi^\mathsf{R})^{k-w_0}  + \sum_{k=w_1}^{\infty} k \pi_1^{X} \psi^\mathsf{R} (1 - \psi^\mathsf{R})^{k-w_1}\) .

    Combining the first and second summations, and using the result \(\sum_{k=s}^{\infty} k \theta^{k-s} = {\theta}/{(1 - \theta)^2} + {s}/{(1 - \theta)}\) and simplifying, we obtain
    \begin{align}
        \mathbb{E}[A(n)] &= \left( \frac{1}{\psi^\mathsf{R}} - 1 \right) + w_0 + \pi_1^X (w_1 - w_0).
    \end{align}
    Thus, the long-term time average of the expected AoI under SRP is
    \begin{align}
        \lim_{N \rightarrow \infty} \frac{1}{N} \sum_{n=1}^{N} \mathbb{E}[A(n)] = \left( \frac{1}{\psi^\mathsf{R}} - 1 \right) + w_0 + \pi_1^X (w_1 - w_0).
        \label{eq:E_A}
    \end{align}

\section{Proof of Theorem~\ref{th:lowerbound}}\label{lower_bound_aoi}

    \begin{figure}[t]
        \includegraphics[width=\columnwidth]{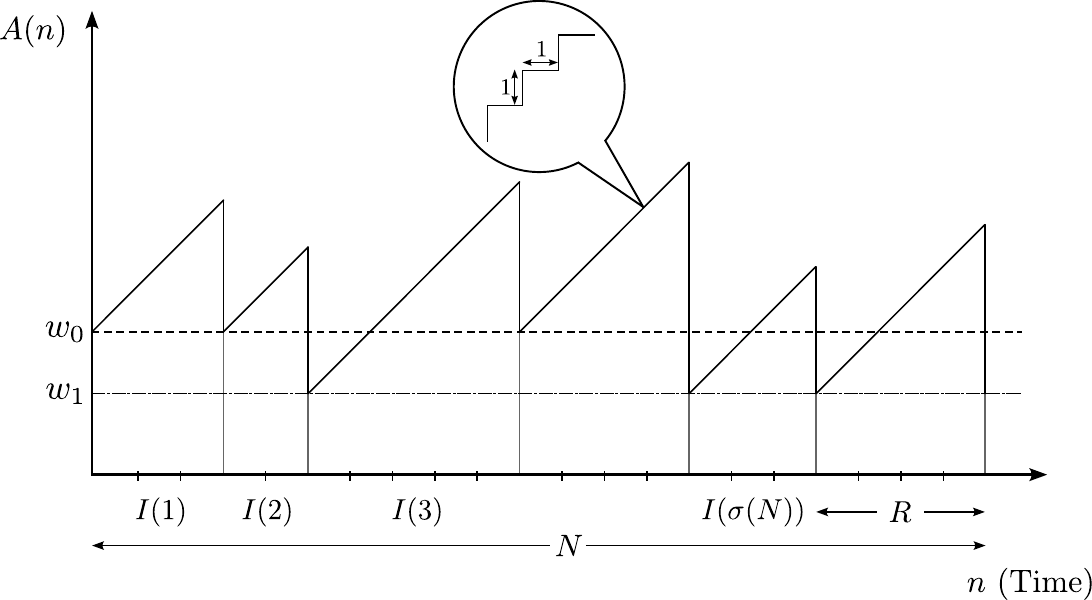}
        \caption{A representative sample path illustrating the time evolution of the Age of Information (AoI) over \( N \) time slots is shown. The variable \( I(k) \) denotes the inter-arrival time between the \((k - 1)^{\text{th}}\) and \(k^{\text{th}}\) packet arrivals. The quantity \( \sigma(N) \) represents the total number of packets received within the observation window of length \( N \).}
        \label{fig:aoi_sample_path}
    \end{figure}

    The proof closely follows the approach of Theorem~1 in \cite{sinhaInfocomThroughput}, with simple yet important modifications, primarily due to the age evolution depending on the estimated state of the monitored environment. Consider a time horizon \( n \in \{1, 2, 3, \dots, N\} \), and define the following variables: \( \sigma(N) \) denotes the total number of successful updates within \( N \) time slots; \( I(m) \) denotes the inter-arrival time between the \( (m-1)^\text{th} \) and \( m^\text{th} \) successful updates, for \( m \in \{1, 2, \dots, \sigma(N)\} \); specifically, \( I(1) \) denotes the number of slots until the first successful update, \( I(2) \) the number of slots between the first and second updates, and so on; \( R \) denotes the number of remaining slots after the last successful update until slot \( N \). (See Fig. \ref{fig:aoi_sample_path}) 

    Additionally, let the binary variable \( s(m) \) indicate whether the \( m^\text{th} \) sample corresponds to State \( 1\) . That is,   
    \begin{align}
        s(m) &\triangleq \begin{cases}
        1, & \text{if the } m^\text{th} \text{ sample corresponds to State \( 1\) }, \\
        0, & \text{otherwise}.
        \end{cases}
    \end{align}
    Considering the monitored process has reached a steady state, the probability of being in either state is governed by the stationary distribution of process \( X\) . By definition of \( s(m) \), it follows that  
    \begin{align}
        s(m) &= 
        \begin{cases}
        1, & \text{with probability } \pi_{1}^{X}, \\
        0, & \text{with probability } 1 - \pi_{1}^{X} = \pi_{0}^{X},
        \end{cases}
    \end{align}
    and 
    \begin{align}
        \mathbb{E}[s(m)]   = \pi_{1}^{X}, \quad \forall\ m. \label{eq:expectation_of_sm}
    \end{align}
    The following holds true:
    \begin{equation}
        N = \sum_{m=1}^{\sigma(N)} I(m) + R. \label{eq:lin_comb}
    \end{equation}    

    For the \( m^\text{th} \) update, during the inter-arrival time between the \( (m-1)^\text{th} \) and \( m^\text{th} \) updates, the AoI evolves as follows: \( w_i, w_i + 1, w_i + 2, \dots, w_i + I(m) - 1 \). Here, \( w_i \) depends on the state estimated from the previous sample. The total accumulated AoI during this inter-arrival period, is given by \( w_i + (w_i + 1) + (w_i + 2) + \dots + (w_i + I(m) - 1)\)  which simplifies to \( (I^2(m) - I(m))/2 + \left(w_0 + s(m{-}1)(w_1 - w_0)\right) I(m)\) .

    Similar to the proof of Theorem~1 in~\cite{sinhaInfocomThroughput}, we consider the sum of AoI over \( N \) slots, \( \sum_{n=1}^{N} A(n) \), which can be expressed in terms of \( I(m) \), the inter-arrival times of updates, \( R \), the number of remaining slots after the final update, \( N \), and \( \sigma(N) \). For a single sample path (see Fig.~\ref{fig:aoi_sample_path}), and by applying Jensen's inequality to the sample mean expressions \( (1 / \sigma(N) )\sum_{m=1}^{\sigma(N)} I(m) \) and \(( 1 / \sigma(N)) \sum_{m=1}^{\sigma(N)} I^2(m) \), which appear in \( \frac{1}{N} \sum_{n=1}^{N} A(n) \), we obtain:
    \begin{multline}\label{eq:L}
        \frac{1}{N}\sum_{n=1}^{N} A(n) 
        \geq \frac{1}{2} \left[\frac{R^2}{N} + \frac{(N - R)^2}{N\sigma(N)} - 1\right] + w_0 \\ 
        + \frac{(w_1 - w_0)}{N} \left(s[\sigma(N)] R + \sum_{m=1}^{\sigma(N)} s(m-1) I(m)\right).
    \end{multline}

    Let the right-hand side of the above equation be denoted as \( L \). So far, we have considered a single sample path of the AoI evolution. Owing to the randomness in the channel state, the number of arrivals, and inter-arrival durations, both \( R \) and \( \sigma(N) \) are random variables for a given \( N \). Taking expectation, we obtain:
    \begin{multline} \label{eq:SbyN}
        \mathbb{E}[L] 
        = \frac{1}{2} \left[\mathbb{E}\left(\frac{R^2}{N}\right) + \mathbb{E}\left(\frac{(N - R)^2}{N\sigma(N)}\right) - 1\right] + w_0\\
        + \frac{w_1 - w_0}{N} \left( \mathbb{E}[s[\sigma(N)] R] + \mathbb{E}\left[\sum_{m=1}^{\sigma(N)} s(m-1) I(m)\right] \right).
    \end{multline}

    We now evaluate the two of the three expectations that appear in \eqref{eq:SbyN}. 
    First,
    \begin{align}
        \mathbb{E}[s[\sigma(N)] \cdot R]  
        &\overset{(a)}{=} \mathbb{E}[s(\sigma(N))] \mathbb{E}[R]\\  
        &\overset{(b)}{=} \pi_{1}^{X} \mathbb{E}[R],
    \end{align}
    where (a) follows from the independence of \( s(\sigma(N)) \) and \( R \), and (b) from homogeneity and~\eqref{eq:expectation_of_sm}. Second,
    \begin{align*}
        &\mathbb{E} \left[\sum_{m=1}^{\sigma(N)} s(m-1) I(m) \right] \nonumber\\
        &\overset{(a)}{=} \mathbb{E}_{\sigma(N)} \left[ \mathbb{E} \left[\sum_{m=1}^{\sigma(N)} s(m-1) I(m) \mid \sigma(N) \right] \right] \\
        &\overset{(b)}{=} \mathbb{E}_{\sigma(N)} \left[ \sum_{m=1}^{\sigma(N)} \mathbb{E} [s(m-1) I(m) \mid \sigma(N)] \right] \\  
        &\overset{(c)}{=} \mathbb{E}_{\sigma(N)} \left[ \sum_{m=1}^{\sigma(N)} \mathbb{E}[s(m-1)] \mathbb{E}[I(m)] \right] \\  
        &\overset{(d)}{=} \mathbb{E}_{\sigma(N)} \left[ \sum_{m=1}^{\sigma(N)} \pi_1^{X} \mathbb{E}[I(m)] \right] \\    
        &\overset{(e)}{=} \pi_1^{X} \mathbb{E}_{\sigma(N)} \left[ \mathbb{E} \left( \sum_{m=1}^{\sigma(N)} I(m) \mid \sigma(N) \right) \right] \\  
        &\overset{(f)}{=} \pi_1^{X} \mathbb{E} \left[ \sum_{m=1}^{\sigma(N)} I(m) \right]  
        \overset{(g)}{=} \pi_1^{X} \mathbb{E}[N - R] 
        \overset{(h)}{=} \pi_1^{X} (N - \mathbb{E}[R]).
    \end{align*}
    Steps (a) and (f) use the law of iterated expectation; (b), (e), and (h) follow from linearity of expectation; (c) from independence of \( s(m-1) \) and \( I(m) \); (d) from~\eqref{eq:expectation_of_sm}; and (g) from~\eqref{eq:lin_comb}.

    Substituting into \eqref{eq:SbyN}, we get:
    \begin{align*}
        \mathbb{E}[L] = \frac{1}{2} \mathbb{E} \left[\frac{R^2}{N} + \frac{(N - R)^2}{N\sigma(N)} - 1\right] + w_0 + \pi_{1}^{X}(w_1 - w_0).
    \end{align*}
    We will now lower bound the right-hand side of the above equation. 

    Note that the smallest value \(R\) can take is \(0\), and we have
    \begin{align}
        \frac{R^2}{N} + \frac{(N - R)^2}{N\sigma(N)} \geq \frac{N}{1 + \sigma(N)}.
    \end{align}
    Subtracting \(1\) from both sides and then taking expectation with respect to the randomness in \(R\) and \(\sigma(N)\), we obtain
    \begin{align}
        \mathbb{E} \left[\frac{R^2}{N} + \frac{(N - R)^2}{N\sigma(N)} - 1\right] \geq \mathbb{E} \left[\frac{N}{1 + \sigma(N)} - 1\right].
    \end{align}
    Thus,
    \begin{align}
        \mathbb{E}[L] \geq \frac{1}{2} \left(N \mathbb{E}\left[\frac{1}{1 + \sigma(N)}\right] - 1\right) + w_0 + \pi_{1}^{X}(w_1 - w_0).
    \end{align}
    Applying Jensen's inequality, we get:
    \begin{align}
        \mathbb{E}[L] \geq \frac{1}{2} \left(\frac{1}{\frac{1}{N} + \frac{1}{N} \mathbb{E}[\sigma(N)]} - 1\right) + w_0 + \pi_{1}^{X}(w_1 - w_0).
    \end{align}

    Noting that:
    \begin{align}
        \lim_{N \to \infty} \frac{1}{N} \mathbb{E}[\sigma(N)] = \lim_{N \to \infty} \frac{1}{N} \sum_{n=1}^{N} \mathbb{E}[d(n)],
    \end{align}
    and using the definition of long-term throughput~\eqref{eq:long_term_throughput}, and from \eqref{eq:L}, we finally obtain:
    \begin{align}
        \bar{A} &= \lim_{N \to \infty} \frac{1}{N} \mathbb{E} \left[\sum_{n=1}^{N} A(n)\right] \nonumber \\ 
        &\geq \frac{1}{2} \left(\frac{1}{\hat{q}} - 1\right) + w_0 + \pi_{1}^{X}(w_1 - w_0).
    \end{align}

\end{appendices}
\balance

\bibliography{references}
\bibliographystyle{IEEEtran}

\end{document}